\documentclass[10pt]{article}
\usepackage{amsmath, amssymb, amsthm, mathrsfs}
\usepackage{enumerate}
\usepackage{colordvi}

\newtheorem{claim}{Claim}[section]
\newtheorem{theorem}[claim]{Theorem}
\newtheorem{remark}[claim]{Remark}
\newtheorem{remarks}[claim]{Remarks}
\newtheorem{corollary}[claim]{Corollary}

\begin{document}
\begin{center}
{\Large{\textbf{Spectral estimates for Dirichlet Laplacians \\ [.01em] and Schr\"odinger operators on geometrically \\ [.3em] nontrivial cusps}}}

\bigskip

{\large{Pavel Exner and Diana Barseghyan}\footnote{The research was partially supported by the Czech Science Foundation within the project P203/11/0701. We are obliged to the referee for useful remarks.}}

\bigskip

\emph{To Elliott Lieb on the occasion of his 80th birthday.}

\end{center}

\bigskip

\textbf{Abstract.} The goal of this paper is to derive estimates of eigenvalue moments for Dirichlet Laplacians and Schr\"odinger operators in regions having infinite cusps which are geometrically nontrivial being either curved or twisted; we are going to show how those geometric properties enter the eigenvalue bounds. The obtained inequalities reflect the essentially one-dimensional character of the cusps and we give an example showing that in an intermediate energy region they can be much stronger than the usual semiclassical bounds.

\bigskip


\bigskip


\section{Introduction} \label{s: intro}
\setcounter{equation}{0}

The object of our interest in this paper will be Schr\"odinger type operators
\begin{equation}
\label{cusp-schroed}
H_\Omega=-\Delta_D^\Omega-V
\end{equation}
with a bounded measurable potential $V\ge 0$ on $L^2(\Omega)$,
where $-\Delta_D^\Omega$ is the Dirichlet Laplacian on a region $\Omega\subset \mathbb{R}^d$. We will be particularly interested is situations where $\Omega$ is unbounded but $H_\Omega$ still has a purely discrete spectrum.

This is not the case, of course, for most open regions; a necessary condition is the quasi-boundedness of $\Omega$ which means the requirement \cite{AF03}
$$
\lim_{\stackrel{x\in\Omega}{|x|\rightarrow\infty}} \mathrm{dist}(x,\partial\Omega)=0\,.
$$
Nevertheless, it is well known --- see \cite{Si83} or \cite{GW11} and references therein --- that for some unbounded regions the spectrum may be purely discrete; typically it happens if $\Omega$ has cusps. The negative spectrum of $H_\Omega$ consists of a finite number of eigenvalues counted with their multiplicities. In this situation one can ask about bounds on the negative spectrum moments in terms of their geometrical properties, in the spirit the seminal work of of Lieb and Thirring \cite{LT76}, or in the present context referring to Berezin, Lieb, Li and Yau \cite{Be72a, Be72b, Li73, LY83}. Estimates of this type have been derived recently in \cite{GW11} for various cusped regions; a typical example is $\Omega = \{ (x,y)\in\mathbb{R}^2:\: |xy|<1 \}$ with hyperbolic ends.

Our aim is the present paper is to investigate situations when such infinite cusps of $\Omega$ are geometrically nontrivial being either curved or twisted and to find in which way does the geometry influence the spectral estimates. First we note that if $\Omega$ can be regarded as a union of subsets of a different geometrical nature it could be useful to make a decomposition and to find estimates from those for separate parts using, say, bracketing technique. Motivated by this observation we will consider in this paper always a single or double cusp-shaped region $\Omega$.

Our strategy in this paper is to employ appropriate coordinate transformations to rephrase the problem as spectral analysis of Schr\"odinger operators on geometrically simpler cusped regions and to deal with the latter using the method proposed in \cite{LW00}. It is naturally not the only possibility. One can apply the reduction procedure to the curved and twisted cusps directly, similarly as it was done for bulged tubes and boundary perturbations in \cite{ELW04}. We do not follow this route here and thus we cannot compare our bounds to those one would obtain in this way; we limit ourselves to the observation that the present approach has the advantage of producing relatively simple bounds in terms of globally defined quantities such as curvature, cusp radius, etc., also in cases when the regions in question have complicated geometry --- e.g., multiple sharp bends or knots in higher dimensions --- when a direct reduction would be cumbersome indeed.

We will start from discussing the simplest case of a curved planar cusp and derive estimates on negative spectrum moments which include a curvature-induced potential describing the effective attractive interaction coming from the region geometry. After this motivating considerations we are going to proceed to generalization to curved cusps in  $\mathbb{R}^d,\: d\ge 2$; before doing that we shall present in Section~\ref{s: BLY} an example showing that for regions with finitely cut cusps the obtained inequality can be at intermediate energies much stronger than the usual estimate using phase-space volume. In Section~\ref{s: twisted} we will consider cusps of a non-circular cross section in $\mathbb{R}^3$ which are straight but twisted. The geometry of the region will be again involved in the obtained eigenvalue estimates, now in a different way than for curved cusps, because the effective interaction associated with twisting is repulsive rather than attractive\footnote{The repulsive character of the effective interaction coming from twisting was noticed first in \cite{CB96}; a rigorous way to express this fact can be stated in terms of the appropriate Hardy-type inequality \cite{EKK08}.}.

\section{A warm-up: curved planar cusps} \label{s: twodim}
\setcounter{equation}{0}

We start with an unbounded cusp-shaped $\Omega\subset \mathbb{R}^2$ assuming that its boundary is sufficiently smooth, to be specified below. To describe the region $\Omega$ we fix first a curve regarded as its axis and employ the natural locally orthogonal coordinates in its vicinity, in analogy with the theory of quantum waveguides \cite{ES89}, which will be used to ``straighten'' the cusp translating its geometric properties into those of the coefficients of the resulting operator.

To be specific, we characterize our region by three functions: sufficiently smooth $a,b:\: \mathbb{R} \to \mathbb{R}^2$ and a positive continuous $f:\: \mathbb{R} \to \mathbb{R}^+$, in such a way that
 \begin{equation} \label{Omega-2}
\Omega :=\{ (a(s)-u \dot b(s),b(s)+u \dot a(s)):\: s\in\mathbb{R},\, |u|<f(s)\, \}\,,
 \end{equation}
where dot marks the derivative with respect to $s\,$; to make the region $\Omega$ cusp-shaped we shall always suppose that
 \begin{equation} \label{Omega-cusp}
\lim_{|s|\to\infty} f(s)=0\,.
 \end{equation}
Dealing with a single cusp would mean, of course, to consider
 \begin{equation} \label{Omega-2+}
\Omega^+ :=\{ (a(s)-u \dot b(s),b(s)+u \dot a(s)) \in \Omega :\: s\ge 0\, \}\,,
 \end{equation}
or any other semi-infinite interval of the longitudinal variable, but it is convenient to treat first the case with two cusp-shaped ends.

Since the reference curve $\Gamma=\{ (a(s),b(s)):\:  s\in\mathbb{R}\}$ can be always parametrized by its arc length we may suppose without loss of generality that $\dot a(s)^2+\dot b(s)^2=1$ and $s$ is the arc length. The signed curvature $\gamma(s)$ of $\Gamma$ is then given by
 $$ 
\gamma(s)=\dot b(s)\ddot a(s)-\dot a(s)\ddot b(s)\,,
 $$ 
and the region (\ref{Omega-2}) is \emph{de facto} determined by the functions $\gamma$ and $f$ only because the Cartesian coordinates of the points of $\Gamma$ can be obtained from
\begin{eqnarray*}
a(s)=a(s_0)+\int_{s_0}^s\cos\left(\int_{s_0}^t \gamma(\xi)\,\mathrm{d}\xi\right)\,\mathrm{d}t,
\\
b(s)=b(s_0)+\int_{s_0}^s\sin\left(\int_{s_0}^t \gamma(\xi)\,\mathrm{d}\xi\right)\,\mathrm{d}t
\end{eqnarray*}
with a fixed point $s_0$, possibly modulo Euclidean transformations of the plane.

The characterization of the region (\ref{Omega-2}) through the curvilinear coordinate makes sense only if the latter can be uniquely defined which imposes two different restrictions. First of all, the transverse size must not be too large, the inequality $|f(s)\gamma(s)|<1$ must hold at any point of the curve. In addition, we have to assume that the used coordinate system makes sense globally, i.e. that the map $(s,u) \mapsto (a(s)-u \dot b(s),b(s)+u \dot a(s))$ is injective, in other words, that the region $\Omega$ does not intersect itself. Note that the last hypothesis can be relaxed if we consider (\ref{cusp-schroed}) not as an operator in $L^2(\mathbb{R}^2)$ but instead acting on $L^2$ functions in an appropriate covering space of the plane.

Under the condition (\ref{Omega-cusp}) the region is quasi-bounded so it may have a purely discrete spectrum. It is indeed the case; recall that the necessary and sufficient condition for the purely discrete spectral character \cite[Theorem~2.8]{BS72} is that we can cover $\Omega$ by a family of unit balls whose centres tend to infinity in such a way that the volumes of their intersections with $\Omega$ tend to zero; it is not difficult to construct such a ball sequence if (\ref{Omega-cusp}) is valid.

Our main aim here is to prove bounds on the eigenvalue moments; as usual when dealing with a Lieb-Thirring type problem we restrict our attention to the negative part of the spectrum noting that it can be always made non-empty by including a suitable constant into the potential.
\begin{theorem} \label{thm: 2D bound}
Consider the Schr\"odinger operator (\ref{cusp-schroed}) on the region (\ref{Omega-2}). Suppose that the curvature $\gamma\in C^4$, the inequality  $\|f(\cdot)\gamma(\cdot) \|_{L^\infty(\mathbb{R})}<1$ holds true, and $\Omega$ does not intersect itself. Then for any $\sigma\ge3/2$ we have the estimate\footnote{The positive and negative part of a quantity $x$ are conventionally defined as $x_\pm = \frac12(|x|\pm x)$.}
\begin{eqnarray}
\nonumber
\mathrm{tr}\left(H_\Omega\right)_-^\sigma
\le\left\|1+f|\gamma|\right\|_\infty^{-2\sigma}\,L_{\sigma,1}
^{\mathrm{cl}}\int_{\mathbb{R}}\sum_{j=1}^\infty\biggl(-\left(\frac{\pi j}{2f(s)}
\right)^2 +\|1+f|\gamma|\|_\infty^2 W^-(s)
\\
\label{eq2.4}
+\|1+f|\gamma|\|_\infty^2 \|\widetilde{V}(s,\cdot)\|_\infty\biggr)_+^{\sigma+1/2}\,\mathrm{d}s\,,
\end{eqnarray}
where $\|\cdot\|_\infty \equiv \|\cdot\|_{L^\infty(\mathbb{R})}$ and
$L_{\sigma,1}^{\mathrm{cl}}$ is the known semiclassical constant,
 \begin{equation} \label{LTconstant}
L_{\sigma,1}^{\mathrm{cl}} := \frac{\Gamma(\sigma+1)}{\sqrt{4\pi} \Gamma(\sigma+\frac32)}\,,
 \end{equation}
and furthermore, we have introduced
$$
W^-(s) :=\frac{\gamma(s)^2}{4\left(1-f(s)|\gamma(s)|\right)^2} +\frac{f(s)|\ddot\gamma(s)|}{2\left(1-f(s)|\gamma(s)|
\right)^3} +\frac{5f^2(s)\dot\gamma(s)^2}{4\left(1-f(s)|\gamma(s)|\right)^4}
$$
and $\widetilde{V}(s,u):=V\big(a(s)-u \dot b(s),b(s)+u\dot a(s)\big)$.
 \end{theorem}
\begin{proof}[Proof:]
Using the standard ``straightening'' transformation \cite{ES89} we infer that $H_\Omega$ is unitarily equivalent to the operator $H_0$ on $L^2(\Omega_0)$ acting as
$$
(H_0\psi)(s,u) = -\frac{\partial}{\partial s}\left(\frac{1}{(1+u\gamma(s))^2}\frac{\partial\psi}{\partial s}(s,u)\right)-\frac{\partial^2\psi}{\partial u^2}(s,u)+((W-\widetilde{V})\psi)(s,u)\,,
$$
where $\Omega_0=\{(s,u):\: s\in\mathbb{R},\, |u|<f(s)\}$, the curvature-induced potential is
$$
W(s,u):= -\frac{\gamma^2(s)}{4(1+u\gamma(s))^2}+\frac{u\ddot\gamma(s)}{2(1+u\gamma(s))^3}-\frac{5}{4}
\frac{u^2\dot\gamma^2(s)}{(1+u\gamma(s))^4}
$$
and Dirichlet boundary conditions are imposed at $u=\pm f(s)$. In view of the unitary equivalence it is enough to establish inequality (\ref{eq2.4}) for the operator $H_0$.

Next we employ a simple minimax-principle estimate. Consider the operator $H_0^-$ defined on the domain $\mathcal{H}^2_0 (\Omega_0)$ in $L^2(\Omega_0)$ by
$$
H_0^-=-\Delta_D^{\Omega_0} -\left\|1+f|\gamma|\right\|_\infty^2(W^- +\widetilde{V})\,,
$$
where $-\Delta_D^{\Omega_0}$ is as usual the corresponding Dirichlet Laplacian; it is obvious that
 \begin{equation}
 \label{eq2.5}
H_0\ge\left\|1+f|\gamma|\right\|_\infty^{-2}\,H_0^-
 \end{equation}
holds true, hence we may get the desired result by establishing a bound on the trace of the operator $(H_0^-)^\sigma$.

We take inspiration from \cite{Wei08} and use a variational argument to reduce the problem to a Lieb-Thirring inequality with operator-valued potential. Given a function $g\in\,C_0^\infty (\Omega_0)$ we can write
\begin{eqnarray*}
\lefteqn{\|\nabla\,g\|^2_{L^2(\Omega_0)}
-\left\|1+f|\gamma|\right\|_\infty^2\int_{\Omega_0}
\left(W^- +\widetilde{V}\right)(s,u)\, |g(s,u)|^2\,\mathrm{d}s\,\mathrm{d}u} \\ && =\int_{\Omega_0}\left|\frac{\partial g}{\partial s}(s,u)\right|^2\,\mathrm{d}s\,\mathrm{d}u
+\int_{\mathbb{R}}\,\mathrm{d}s \int_{-f(s)}^{f(s)}\bigg(\left|\frac{\partial g}{\partial u}(s,u)\right|^2 \\ &&
\quad -\left\|1+f|\gamma|\right\|_\infty^2 \left(W^- +\widetilde{V}\right)(s,u)\, |g(s,u)|^2\bigg)\,\mathrm{d}u
\\ &&
\ge\int_{\Omega_0}\left|\frac{\partial g}{\partial s}(s,u)\right|^2\,\mathrm{d}s\,\mathrm{d}u +\int_{\mathbb{R}}\left\langle\,
H\left(s,\widetilde{V},W^-\right)g(s,\cdot), g(s,\cdot)\right\rangle_{L^2(-f(s),f(s))}
\,\mathrm{d}s\,,
\end{eqnarray*}
where $H(s,\widetilde{V},W^-)$ is the negative part of Sturm-Liouville operator
$$
-\frac{\mathrm{d}^2}{\mathrm{d}u^2} -\left\|1+f|\gamma|\right\|_\infty^2
\left(W^- +\widetilde{V}\right)
$$
defined on $C_0^\infty(-f(s),f(s))$ with Dirichlet conditions at $u=\pm f(s)$.

Next we introduce the complement $\widehat{\Omega}_0 :=\mathbb{R}^2 \backslash\overline{\Omega}_0$ to the straightened region $\Omega_0$ and consider functions of the form $h=g+v$ with $g\in C_0^\infty(\Omega_0)$ and $v\in C_0^\infty(\widehat{\Omega}_0)$ which we may regard as functions in $\mathbb{R}^2$ extending them by zero to $\widehat{\Omega}_0$ and $\Omega_0$, respectively. Similarly we extend $H(s,\widetilde{V},W^-)$ to the operator on $C_0^\infty(\mathbb{R})$ acting as $H(s,\widetilde{V},W^-) \oplus 0$ with the zero component on $C_0^\infty \left(\mathbb{R} \backslash[-f(s),f(s)]\right)$. We have
\begin{eqnarray*}
\lefteqn{\|\nabla\,g\|^2_{L^2(\Omega_0)} +\|\nabla\,v\|^2_{L^2(\widehat{\Omega}_0)} -\left\|1+f|\gamma|\right\|_\infty^2 \int_{\Omega_0}\left(W^- +\widetilde{V}\right)(s,u) |g(s,u)|^2\,\mathrm{d}s\,\mathrm{d}u}
\\ &&
\ge\int_{\mathbb{R}^2}\left|\frac{\partial h}{\partial s}(s,u)\right|^2\,\mathrm{d}s\,\mathrm{d}u +\int_{\mathbb{R}}\left\langle
\,H(s,\widetilde{V},W^-)\,
h(s,\cdot),\,h(s,\cdot)\right\rangle_{L^2(\mathbb{R})}
\,\mathrm{d}s\,. \phantom{AAAA}
\end{eqnarray*}
This inequality holds true for any function $g\in C_0^\infty\left(\mathbb{R}^2\backslash\partial\Omega_0\right)$  and its left-hand side is the quadratic form corresponding to the operator $H_0^- \oplus \left(-\Delta^{\widehat{ \Omega}_0}_D\right)$, while the right-hand one is the form associated with the operator
$$
-\frac{\partial^2}{\partial\,s^2}\otimes\,I_{L^2(\mathbb{R})} +H\left(s,\widetilde{V},W^-\right)
$$
defined on the larger domain $\mathcal{H}^1\left(\mathbb{R}, L^2(\mathbb{R})\right)$. Since  $-\Delta^{ \widehat{\Omega}_0}_D$ is positive, we infer from the minimax principle that
$$
\mathrm{tr}\,(H_0^-)_-^\sigma \le \,\mathrm{tr}\left(-\frac{\partial^2}{\partial\,s^2}\otimes \,I_{L^2(\mathbb{R})}+ H(s,\widetilde{V},W^-)\right)_-^\sigma
$$
holds for any nonnegative number $\sigma$. This makes it possible to employ the version of Lieb-Thirring inequality for operator-valued potentials [LW00] for operator valued potentials which yields
\begin{equation}
\label{eq2.7}
\mathrm{tr}\,(H_0^-)_-^\sigma\le\,L_{\sigma,1}^{\mathrm{cl}} \int_{\mathbb{R}}\mathrm{tr}\,\left(H\left(s,\widetilde{V},W^-\right)
\right)_-^{\sigma+1/2}\, \mathrm{d}s\,,\;\;\sigma\ge3/2\,,
\end{equation}
with the semiclassical constant $L_{\sigma,1}^{\mathrm{cl}}$.
Define next Sturm-Liouville operator $L_f(s)$ on $L^2\left(-f(s), f(s)\right)$ acting as
$$
L(s,\widetilde{V},W^-)=-\frac{\mathrm{d}^2}{\mathrm{d}u^2}-\left\|1+f|\gamma|\right\|_\infty^2
\biggl(W^-(s) +\|\widetilde{V}(s,\cdot)\|_\infty\biggr)
$$
with Dirichlet conditions at $u=\pm f(s)$. In view of the last two inequalities and the minimax principle we find that
\begin{equation}
\label{spectrum}
\mathrm{tr}\,(H_0^-)_-^\sigma\le\,L_{\sigma,1}^{\mathrm{cl}}
\int_{\mathbb{R}} \mathrm{tr}\left(L_f\left(s,\widetilde{V},
W^-\right)\right)_-^{\sigma+1/2}.
\end{equation}
holds for any $\sigma\ge3/2$. Now we need  $-\mu_j(s),\, j=1,2,\ldots$, being the negative eigenvalues of $L_f(s,\widetilde{V},W^-)$ which can be easily found using the fact $L(s,0,0)$ as the Dirichlet Laplacian on $(-f(s),f(s))$ has eigenvalues $\big(\frac{\pi j}{2f(s)}\big)^2,\: j=1,2,\dots\,,$ and the potential is independent of $u$. This makes it possible to evaluate the right-hand side (\ref{spectrum}) and using the estimate (\ref{eq2.5}) we establish for any $\sigma\ge3/2$ the inequality
\begin{eqnarray*}
\mathrm{tr}\left(H_0\right)_-^\sigma
\le\left\|1+f|\gamma|\right\|_\infty^{-2\sigma}\,L_{\sigma,1}
^{\mathrm{cl}}\int_{\mathbb{R}}\sum_{j=1}^\infty\biggl(-\left(\frac{\pi j}{2f(s)}
\right)^2 +\|1+f|\gamma|\|_\infty^2 W^-(s)
\\
\nonumber
+\left\|1+f|\gamma|\right\|_\infty^2\|\widetilde{V}(s,\cdot)\|_\infty\biggr)_+^{\sigma+1/2}\,\mathrm{d}s\,,
\end{eqnarray*}
and the unitary equivalence between $H_\Omega$ and $H_0$ yields the sought claim.
\end{proof}

\begin{corollary} \label{2cor}
Consider the operator $H_{\Omega^+}$ on the region (\ref{Omega-2+}) defined in analogy with (\ref{cusp-schroed}) with Dirichlet condition at $s=0$. The inequality (\ref{eq2.4}) holds again with integration variable running now over the interval $(0,\infty)$.
\end{corollary}
\begin{proof}
The claim follows by a bracketing argument. Imposing Dirichlet condition at the segment $\{ (s,u):\: s=0,\, |u|<f(0)\}$ we get $H_\Omega \le H_{\Omega^+} \oplus H_{\Omega^-}$ with the obvious consequence for the negative eigenvalues of these operators.
\end{proof}

\begin{remarks} \label{thickend}
{\rm (a) Note that we have not used the condition (\ref{Omega-cusp}). If it is not satisfied, the spectrum of $H_\Omega$ may not be purely discrete; the proved inequality remains valid as long as we stay below $\inf\sigma_\mathrm{ess} (H_\Omega)$. It is also useful to notice that the validity of the result can be in analogy with \cite{EW00, ELW04} extended to any $\sigma\ge 1/2$; the price we have to pay is only a change in the constant, $L_{\sigma,1}^{\mathrm{cl}}$ being replaced with $r(\sigma,1) L_{\sigma,1}^{\mathrm{cl}}$ with the factor $r(\sigma,1)\le 2$ if $\sigma<3/2$. \\
(b) If (\ref{Omega-cusp}) is satisfied the argument leading (\ref{eq2.4}) shows, in particular, that the negative spectrum of $H_D^\Omega$ is discrete and since we may add an arbitrary constant to the potential, it gives us an alternative way to demonstrate the purely discrete character of the spectrum.
}
\end{remarks}

\section{Comparison with phase space estimate}
\label{s: BLY}
\setcounter{equation}{0}

Having different eigenvalue estimate one asks naturally how they do compare. For bounded regions we have a standard estimate, a modification of Lieb-Thirring inequality based on assessment of the phase-space volume, obtained first by Berezin \cite{Be72a, Be72b} and Lieb \cite{Li73}, and in another form later by Li and Yau  \cite{LY83, Wei08}. There is no doubt that it provides an estimate with the correct semiclassical behaviour. Our point is this section is to show that if the region has ``thin'' parts there may exist an intermediate interval of energies where the estimates of the type discussed in the previous section are considerably stronger than the the classical inequality mentioned above. This different spectral behavior is somewhat reminiscent to the two-term Lieb-Thirring estimates derived in \cite{EW00} for Schr\"odinger operators in straight Dirichlet tubes. Similarly as there one has energy regions when the one-dimensional or the multi-dimensional character of the region dominates; the difference is that here the coupling is purely geometrical.

The standard trick to study spectra of Dirichlet Laplacians below certain value is to consider operator (\ref{cusp-schroed}) with a constant, possibly large, potential $-V$ where $V(s,u)=\Lambda>0$, and to look for its negative spectrum. Consider first again the region (\ref{Omega-2}) satisfying the condition (\ref{Omega-cusp}) together with
 \begin{eqnarray}
 \|f\gamma\|_\infty < c < \frac{-\pi-1+\sqrt{(\pi+1)^2 +4\pi}}{2} \approx 0.655\,, \label{fgamma} \\[.5em]
 \max\{ \|f\dot\gamma\|_\infty,\, \|f\ddot\gamma\|_\infty\} < 1\,. \label{fgammader} \phantom{AAAAAAA}
 \end{eqnarray}
Put $W_\Lambda^-(s):=W^-(s)+\Lambda$. In view of the assumption (\ref{fgamma}) and Theorem~\ref{thm: 2D bound} we can estimate $\mathrm{tr}\,(H_\Omega)_-^\sigma$ for any $\sigma\ge3/2$ from above by
 \begin{eqnarray*}
\lefteqn{ \left\|1+f|\gamma|\right\|_\infty^{-2\sigma}\, L_{\sigma,1}^{\mathrm{cl}}\int_{\mathbb{R}}\sum_{j=1}^\infty
\left(-\left(\frac{\pi\,j}{2f(s)}\right)^2 +\left\|1+f|\gamma|
\right\|_\infty^2W_\Lambda^-(s)\right)_+^{\sigma+1/2}\,\mathrm{d}s}
 \\ &&
\le\left\|1+f|\gamma|\right\|_\infty\,L_{\sigma,1}^{\mathrm{cl}}
\int_{f(s)\ge\frac{\pi}{2(1+c)}W_\Lambda^-(s)^{-1/2}} \sum_{j=1}^{\left[2(1+c)f(s) W_\Lambda^-(s)^{1/2}/\pi\right]} W_\Lambda^-(s)^{\sigma+1/2}\,\mathrm{d}s
 \\ &&
\le\frac{2(1+c)}{\pi}\,\left\|1+f|\gamma|\right\|_\infty\,L_{\sigma,1}^{\mathrm{cl}}
\int_{f(s)\ge\frac{\pi}{2(1+c)} W_\Lambda^-(s)^{-1/2}} W_\Lambda^-(s)^{\sigma+1}\,f(s)\,\mathrm{d}s
 \\ &&
\le\frac{8}{\pi}\,L_{\sigma,1}^{\mathrm{cl}}\int_{ f(s)\ge\pi\left(\left(\frac{1+c}{1-c}\right)^2\gamma(s)^2+
4(1+c)^2\Lambda\right)^{-1/2}}
W_\Lambda^-(s)^{\sigma+1}\,f(s)\,\mathrm{d}s\,.
 \end{eqnarray*}
Notice that for a fixed $f$ the right-hand side of the last inequality reaches its maximum if $\gamma(s)f(s)=c$. Consequently, setting $\alpha_c^2:= \frac{\pi^2-c^2(1+c)^2/(1-c)^2}{4(1+c)^2}$ which is a positive number under the assumption (\ref{fgamma}) we get
 \begin{eqnarray}
\lefteqn{\mathrm{tr}\,(H_\Omega)_-^\sigma\le\frac{8}{\pi}\,
L_{\sigma,1}^{\mathrm{cl}} \int_{f(s)\ge\alpha_c\Lambda^{-1/2}} \left(\frac{c^2}{4(1-c)^2f^2(s)}+\Lambda\right)^{\sigma+1}\,f(s)\,\mathrm{d}s}
\nonumber \\ &&
\label{eq2.11}
\le\frac{8}{\pi}\,\left(\frac{c^2}{4(1-c)^2\alpha_c^2}+1\right)^{\sigma+1}\, L_{\sigma,1}^{\mathrm{cl}}\,\Lambda^{\sigma+1} \int_{f(s)\ge\alpha_c\Lambda^{-1/2}} f(s)\,\mathrm{d}s\,; \phantom{AAA}
 \end{eqnarray}
note that the curvature is present in this estimate through the constant $c$ only.

Let us now show that such an estimate can be stronger than the phase-space bound mentioned above which says that the operator  $H_{\Omega'}$ defined by (\ref{cusp-schroed}) on an open bounded region $\Omega'$ with constant potential
$V=\Lambda,\,\Lambda>0$ satisfies
\begin{equation}
\label{eq2.12}
\mathrm{tr}\left(H_{\Omega'}\right)_-^\sigma\le L_{\sigma,1}^{\mathrm{cl}}\Lambda^{\sigma+1}\,
\mathrm{vol}\,(\Omega')\,,\quad \sigma\ge1\,.
\end{equation}
First we shall construct an unbounded cusped region $\Omega$ determined by functions $\gamma$ and
$f$  satisfying the conditions (\ref{Omega-cusp}), (\ref{fgamma}) and (\ref{fgammader}), then we will pass to cut-off regions $\Omega'\subset\Omega$ such that
\begin{equation}
\label{eq2.13}
\mathrm{tr}\left(H_\Omega'\right)_-^\sigma\le \mathrm{tr}\left(H_\Omega\right)_-^\sigma\,,\quad \sigma\ge0\,.
\end{equation}
We choose an arbitrary positive number $\alpha$ and a natural number $N$ and set
 \begin{eqnarray}
f_{\alpha,N}(x): =\frac{\pi}{2}\:x^{-1-\alpha} \qquad\text{for}\quad |x|>N\,, \\
f_{\alpha,N}(x): =\frac{\pi}{2}\:N^{-1-\alpha}\qquad\text{for}\quad |x|\le N\,.
 \end{eqnarray}
Since we have mentioned that the curvature does not play a substantial role in the estimate (\ref{eq2.11}) we put it equal to zero and consider the straight region $\Omega_{\alpha,N} :=\{x\in \mathbb{R}:\: |y|<f_{\alpha,N}(x) \}$.
Then the operator $H_{\Omega_{\alpha,N}}= -\Delta_D^{\Omega_{\alpha,N}} -\Lambda$ satisfies for any $\Lambda\le N^{2(1+\alpha)}$ the following inequality
 \begin{equation}
 \label{eq2.14}
\mathrm{tr}\left(H_{\Omega_{\alpha,N}}\right)_-^\sigma \le\frac{8}{\pi}\,L_{\sigma,1}^{\mathrm{cl}}\Lambda^{\sigma+1}
\int_{f(x)\ge\frac{\pi}{2}\Lambda^{-1/2}} f(s)\,\mathrm{d}x
\le 4\,L_{\sigma,1}^{\mathrm{cl}} \Lambda^{\sigma+1}\,N^{-\alpha}\,.
 \end{equation}
Consider now the finite region $\Omega'_{\alpha,N}:= \{|x|<2^{1/\alpha}N,\,|y|<f_{\alpha,N}(x)\}$ and use Dirichlet bracketing: the left-hand side of the last inequality is not smaller than $\mathrm{tr}\big(H^D_{ \Omega_{\alpha,N}} \big)_-^\sigma$ where the operator has additional Dirichlet conditions imposed at the segments $\left\{x=\pm2^{1/\alpha}N \right\} \cap\Omega_{\alpha,N}$, and this is in turn not smaller than $\mathrm{tr}(H_{\Omega'_{\alpha,N}})_-^\sigma$, in other words,
\begin{equation}
\label{eq2.16}
\mathrm{tr}(H_{\Omega'_{\alpha,N}})_-^\sigma \le 4\,L_{\sigma,1}^{\mathrm{cl}}\Lambda^{\sigma+1}\, N^{-\alpha} \qquad\text{for}\quad \sigma\ge 3/2\,.
\end{equation}
On the other hand, the phase-space estimate (\ref{eq2.12}) gives
$$
\mathrm{tr}\left(H_{\Omega'}\right)_-^\sigma\le 2\pi\,L_{\sigma,1}^{\mathrm{cl}}\,N^{-\alpha}
\,\Lambda^{\sigma+1}\left(\frac{1}{2\alpha}+1\right)\qquad\text{for}\quad \sigma\ge 1\,.
$$
Given $\sigma\ge 3/2$, this can be made much larger than by choosing $\alpha$ small; for $N$ large this difference between the two estimates persists over a large energy interval.

\section{Curved circular cusps in $\mathbb{R}^d$}
\label{s: ddim}
\setcounter{equation}{0}

Let us return now to the subject of Section~\ref{s: twodim} and look how it can be generalized for curved cusps in $\mathbb{R}^d,\,d\ge3$. From the reason which will be given a little later we shall restrict our attention to cusps of a circular cross section. In such a case we will be able to characterize the region as before by specifying the cusp axis and the function determining its radius using curvilinear coordinated as in \cite{CDFK05}.

Let us begin with the region axis. Given the dimension $d\ge3$ we suppose that the axis is a unit-speed $C^{d+2}$-smooth curve $\Gamma:\mathbb{R}\rightarrow\mathbb{R}^d$ which possesses a positively oriented Frenet frame, i.e. a $d$-tuple
$\left\{e_1,\ldots,e_d\right\}$ of functions such that
 \begin{enumerate}[(i)]
 \setlength{\itemsep}{0pt}
 \item $\:e_1=\dot{\Gamma}$,
 \item $\:e_i\in\,C^1(\mathbb{R},\mathbb{R}^d)\,$ holds for any $j=1,\ldots,d$,
 \item $\:\dot{e}_i(s)\,$ lies in the span of $e_1(s),\ldots,e_{i+1}(s)$ for any $j=1,\ldots,d-1$.
 \end{enumerate}
A sufficient condition for existence of such a frame is that the vector values of the derivatives  $\dot{\Gamma}(s), \ddot\Gamma(s),\ldots, \Gamma^{(d-1)}(s)$ are linearly independent for all $s\in\mathbb{R}$; note that this is always satisfied if $d=2$. We have the Frenet-Serret formulae,
$$
\dot{e}_i=\sum_{j=1}^{d-1}\mathcal{K}_{ij}e_j\,,
$$
where $\mathcal{K}_{ij}$ are the entries of the $d\times d$ skew-symmetric matrix of the form
$$
\mathcal{K}=\left(\begin{array}{cccc}
0 &\kappa_1 &\ldots &0
\\
-\kappa_1 &\ldots &\ldots &0
\\
\ldots &\ldots &\ldots & \kappa_{d-1}
\\
0 &\ldots &-\kappa_{d-1} &0
\end{array}\right),
$$
where $\kappa_i:\mathbb{R}\rightarrow\mathbb{R}$ is called the $i$-th curvature of $\Gamma$. Under the assumptions (i)--(iii)  the curvatures are continuous functions of the arc-length parameter $s\in\mathbb{R}$. Consider next a $(d-1)\times(d-1)$ matrix function $\mathcal{R}= (\mathcal{R}_{\mu,\nu})$ determined by the system of differential equations
$$
\dot{\mathcal{R}}_{\mu\nu}+\sum_{\rho=2}^d \mathcal{R}_{\mu,\rho}\mathcal{K}_{\rho,\nu}=0\,,\quad\, \mu,\nu=2,\ldots,d\,,
$$
with the initial conditions at a given point $s_0\in\mathbb{R}$ meaning that $\mathcal{R}(s_0)$ is a rotation matrix in $\mathbb{R}^{d-1}$, i.e. it satisfies the requirements
$$
\det \mathcal{R}(s_0)=1 \qquad\text{and}\qquad \sum_{\rho=2}^d\mathcal{R}_{\mu,\rho}(s_0)
\mathcal{R}_{\nu,\rho}(s_0)=\delta_{\mu,\nu}\,.
$$
Next we associate with $\mathcal{R}(\cdot)$ a $d\times d$ matrix function given by
$$
(\mathcal{R}_{ij}(s)):=\left(\begin{array}{cccc}
1 &0
\\
0 &(\mathcal{R}_{\mu,\nu}(s))
\end{array}\right)
$$
and define the moving frame $\{\widetilde{e}_1, \ldots, \widetilde{e}_d\}\subset \mathbb{R}^d$ along the curve $\Gamma$ by
\begin{equation}
\label{eq3.1}
\widetilde{e}_i:=\sum_{j=1}^d\mathcal{R}_{ij}e_j\,;
\end{equation}
we call it the Tang frame (relative to the given Frenet frame).

Let us pass to the description of the region. Given the Tang frame we can characterize points in the vicinity of $\Gamma$ by means of the corresponding Cartesian coordinates $u_2,\ldots,u_d$ in the normal plane to $\Gamma$ at each point of the curve,
$$
x(s,u_2,\ldots,u_d) :=\Gamma(s)+\sum_{\nu=2}^d\widetilde{e}_\mu(s)u_\mu\,,
$$
in particular, $|u|=\left(\sum_{\nu=2}^d\,u_\nu^2\right)^{1/2}$ measures the radial distance from $\Gamma$. Given a positive function $f:\mathbb{R} \to\mathbb{R}_+$ satisfying the condition (\ref{Omega-cusp}) we define the region
 \begin{equation} \label{Omega-d}
\Omega :=\{ x(s,u_2,\ldots,u_d):\: s\in\mathbb{R},\, |u|<f(s)\, \}\,;
 \end{equation}
in full analogy with the case $d=2$ we assume that the description of $\Omega$ in terms of the curvilinear coordinates makes sense, namely that the radius $f(s)$ is not too large --- to be specified below --- and that the map $(s,u_2,\ldots,u_d) \mapsto x(s,u_2,\ldots,u_d)$ is injective. In a similar way one defines one sided $\Omega^+$ analogous to (\ref{Omega-2+}).

Having described the cusped region we can pass to the operator (\ref{cusp-schroed}) on $L^2(\Omega)$ with a bounded measurable $V\ge0$. Using the again the natural unitary equivalence, as formulated in \cite{CDFK05} for tubes of constant cross section, consisting in rewriting $H_\Omega$ in terms of the curvilinear coordinates and removing the corresponding Jacobian, we pass to the operator on $L^2(\Omega_0)$, where $\Omega_0$ is the straightened region,  $\Omega_0=\left\{ (s,u_1,\ldots,u_{d-1}):\: s\in\mathbb{R},\,|u|<f(s)\right\}$, acting as
\begin{equation}
\label{eq3.4}
H_0=-\partial_1\frac{1}{h^2}\partial_1 -\sum_{\mu=2}^d\partial_\mu^2+W-\widetilde{V}
\end{equation}
with Dirichlet condition at the boundary of the disc, $|u|=f(s)$. Here $\partial_1,\, \partial_\mu$ are the usual shorthands for $\frac{\partial}{\partial s}$ and $\frac{\partial}{\partial u_\mu}$, respectively. Furthermore, $\widetilde{V}(s,u_1,\ldots,u_{d-1}) := V(x(s,u_1,\ldots,u_{d-1}))$ and the curvature-induced part of the potential equals
 $$ 
W:=-\frac{1}{4}\frac{\kappa_1^2}{h^2} +\frac{1}{2}\frac{ h_{11}}{h^3}-\frac{5}{4}\frac{h_1^2}
{h^4}\,,
 $$ 
where
$$
h(s,u_2,\ldots,u_d):= 1-\kappa_1(s)\sum_{\mu=2}^d\mathcal{R}_{\mu2}(s)u_\mu
$$
and the derivatives with respect to $s$ are given explicitly by
\begin{eqnarray*}
&& h_1(\cdot,u)= \sum_{\mu,\alpha=2}^d u_\mu\mathcal{R}_{\mu,\alpha}\dot{\mathcal{K}}_{\alpha,1}-
\sum_{\stackrel{\mu,\alpha=2,\ldots,d}{\beta=1,\ldots,d}}
u_\mu\mathcal{R}_{\mu,\alpha} \mathcal{K}_{\alpha,\beta}\mathcal{K}_{\beta,1}\,,
\\ && h_{11}(\cdot,u)=\sum_{\mu,\alpha=2}^d u_\mu\mathcal{R}_{\mu,\alpha}\ddot{\mathcal{K}}_{\alpha,1}
-\sum_{\stackrel{\mu,\alpha=2,\ldots,d}{\beta=1,\ldots,d}} u_\mu\mathcal{R}_{\mu,\alpha}\left(\dot{\mathcal{K}}_{\alpha,\beta}
\mathcal{K}_{\beta,1} +2\mathcal{K}_{\alpha,\beta}\dot{\mathcal{K}}_{\beta,1}\right)
\\ && \hspace{5em}
+\sum_{\stackrel{\mu,\alpha=2,\ldots,d}{\beta,\gamma=1,\ldots,d}} u_\mu\mathcal{R}_{\mu,\alpha}\mathcal{K}_{\alpha,\beta}\mathcal{K}
_{\beta,\gamma}\mathcal{K}_{\gamma,1}\,.
\end{eqnarray*}
What is crucial in this construction is that we have passed to $H_0$ using the Tang frame (\ref{eq3.1}) because this choice of curvilinear coordinates guarantees that the transformed operator does not contained terms mixing derivatives with respect to the longitudinal and transverse variables.

As in the two-dimensional case our strategy is to estimate the contribution of $W(s,u)$ by a function depending on the longitudinal variable only; we introduce
\begin{eqnarray*}
\lefteqn{W^-(s):=\frac{\kappa_1^2(s)}{4\left(1-f(s)|\kappa_1(s)|\right)^2}
+\frac{f(s)}{2\left(1-f(s)|\kappa_1(s)|\right)^3}\left(\sum_{\mu=2}^d \left(\sum_{\alpha=2}^d\left|\mathcal{R}_{\mu,\alpha}\ddot{\mathcal{K}}
_{\alpha1}\right|\right)^2\right)^{1/2}}
\\ &&
\hspace{-1em} +\frac{f(s)}{2\left(1-f(s)|\kappa_1(s)|\right)^3} \left(\sum_{\mu=2}^d\left(\sum_{\stackrel{\alpha =2,\ldots,d}{\beta=1,\ldots,d}}
\left|\mathcal{R}_{\mu,\alpha}\right| \left(\left|\dot{\mathcal{K}}_{\alpha,\beta} \mathcal{K}_{\beta,1}\right|
+2\left|\mathcal{K}_{\alpha,\beta} \dot{\mathcal{K}}_{\beta,1}\right|\right)\right)^2\right)^{1/2}
\\ &&
+\frac{f(s)}{2\left(1-f(s)|\kappa_1(s)|\right)^3}\left(\sum_{\mu=2}^d
\left(\sum_{\stackrel{\alpha=2,\ldots,d}{\beta,\gamma=1,\ldots,d}}
\left|\mathcal{R}_{\mu,\alpha}\mathcal{K}_{\alpha,\beta}\mathcal{K}_{\beta,\gamma}
\mathcal{K}_{\gamma,1}\right|\right)^2\right)^{1/2}
\\ &&
+\frac{5f(s)}{4\left(1-f(s)|\kappa_1(s)|\right)^4} \left(\sum_{\mu=2}^d\left(\sum_{\alpha=2}^d
\left|\mathcal{R}_{\mu,\alpha} \dot{\mathcal{K}}_{\alpha,1}\right|\right)^2\right)^{1/2}
\\ &&
+\frac{5f(s)}{4\left(1-f(s)|\kappa_1(s)|\right)^4} \left(\sum_{\mu=2}^d\left(\sum_{\stackrel{\alpha=
2,\ldots,d}{\beta=1,\ldots,d}} \left|\mathcal{R}_{\mu,\alpha}\mathcal{K}_{\alpha,\beta}
\mathcal{K}_{\beta,1}\right|\right)^2\right)^{1/2}.
\end{eqnarray*}

Now we are prepared to make the following claim:
\begin{theorem} \label{thm: gen bound}
Let the operator $H_\Omega$ given by (\ref{cusp-schroed}) with a bounded measurable $V\ge0$ correspond to a region $\Omega$ which is not self-intersecting; we assume that it is determined by a  $C^{d+2}$-smooth curve $\Gamma$ and a function $f$ satisfying the condition $\|\kappa_1(\cdot) f(\cdot)\|_\infty<1$, where $\kappa_1$ is the first curvature of $\Gamma$. Then for the negative spectrum of $H_\Omega$ the following inequality holds true,
\begin{eqnarray*}
\lefteqn{\mathrm{tr}\,(H_\Omega)_-^\sigma \le\left\|1+f|\kappa_1|\right\|_\infty^{-2\sigma}\,
L_{\sigma,1}^{\mathrm{cl}} \int_{\mathbb{R}} \sum_{k,m=0,1,\ldots}\biggl(- \left(\frac{j_{k+(d-3)/2,m}}{f(s)}\right)^2}
\\ && \hspace{3em}
+\left\|1+f|\kappa_1|\right\|_\infty^2\biggl(W^-(s) +\|\widetilde{V}(s,\cdot)\|_\infty\biggr)\biggr)_+^{\sigma+1/2} \,\mathrm{d}s\,, \phantom{AAAAAAAAA}
\end{eqnarray*}
where $L_{\sigma,1}^{\mathrm{cl}}$ is the constant (\ref{LTconstant}), the functions $W^-$ and $\widetilde{V}$ have been defined above, and $j_{l,m}$ is the $m$-th positive zero of the first-kind Bessel function $J_l$.
\end{theorem}
\begin{proof}[Proof:]
Using the mentioned unitary equivalence it suffices to establish the claim for the operator $H_0$. Since
$\sum_{\mu=2}^d\mathcal{R}_{\mu2}^2=1$, cf.~\cite{CDFK05}, one can estimate $H_0$ with the help of the operator
$$
H_0^-=-\sum_{\mu=1}^d\partial_\mu^2 -\left\|1+f|\kappa_1|\right\|_\infty^2
\left(W^-+\widetilde{V}\right)
$$
in the following way
\begin{equation}
\label{eq3.8}
H_0\ge\left\|1+f|\kappa_1|\right\|_\infty^{-2}H_0^-\,.
\end{equation}
We follow the same route as in the two-dimensional case starting with the estimate
\begin{eqnarray*}
\lefteqn{\|\nabla\,g\|^2_{L^2(\Omega_0)} -\left\|1+f|\kappa_1|\right\|_\infty^2
\int_{\Omega_0}\, (W^-+\widetilde{V})(s,u) g(s,u)\|^2\:\mathrm{d}s\, \mathrm{d}u_2\ldots\,\mathrm{d}u_d}
\\ && \hspace{-1em}
\ge\int_{\Omega_0}\left|\frac{\partial g}{\partial s}(s,u)\right|^2\,\mathrm{d}s\,\mathrm{d}u_2\ldots\,\mathrm{d}u_d+\int_{\mathbb{R}}\left
\langle\,H(s,\widetilde{V},W^-)\,g(s,\cdot)
,\,g(s,\cdot)\right\rangle_{L^2\left(D_{f(s)}\right)}\,\mathrm{d}s
\end{eqnarray*}
for $g\in\,C_0^\infty(\Omega_0)$, where $u$ is a shorthand for $(u_2,\ldots,u_d)$ and $D_\varrho:=\{ u:\: |u|<\varrho\}$, and furthermore, $H(s,\widetilde{V},W^-)$ is the negative part of the $(d-1)$-dimensional Schr\"{o}dinger operator
$$
-\partial^2_2-\ldots-\partial_d^2 -\left\|1+f|\kappa_1|\right
\|_{\infty}^2\left(W^-+\widetilde{V}\right)
$$
in the disc $D_{f(s)}$ with Dirichlet conditions on its boundary. Next we extend the operator to the complement of $\Omega_0$ in $\mathbb{R}^d$ using $H(s,V^-,\widetilde{W})$ defined on the larger domain $H^1\left(\mathbb{R}, L^2(\mathbb{R}^{d-1})\right)$ and acting as zero of function supported outside $D_{f(s)}$; in analogy with the two-dimensional case this leads to the estimate
$$
\mathrm{tr}\,(H^-_0)_-^\sigma \le\,tr\left(-\partial_s^2 \otimes\,I_{L^2(\mathbb{R}^{d-1})}+
H(s,\widetilde{V},W^-)\right)_-^\sigma\,,\quad\sigma\ge0\,,
$$
and Lieb-Thirring inequality for operator valued-potentials gives
 $$
\mathrm{tr}\,(H^-_0)_-^\sigma\le\, L_{\sigma,1}^{\mathrm{cl}} \int_{\mathbb{R}}\mathrm{tr}\left(H\left(s,\widetilde{V},W^-\right)
\right)_-^{\sigma+1/2}\,\mathrm{d}s\,,\quad\sigma\ge3/2\,.
 $$
To estimate the right-hand side we introduce the following operator in $L^2(D_{f(s)})$,
$$
L\left(s,\widetilde{V},W^-\right): =-\Delta_D^{D_{f(s)}}
-\left\|1+f|\kappa_1|\right\|_{\infty}^2 \left(W^-(s)+\|\widetilde{V}(s,\cdot)\|_\infty\right),
$$
where $\Delta_D^{D_{f(s)}}$ is the corresponding Dirichlet Laplacian; the minimax principle then yields
\begin{equation}
\label{eq3.11}
\mathrm{tr}\,(H^-_0)_-^\sigma\le\, L_{\sigma,1}^{\mathrm{cl}}\int_{\mathbb{R}} \mathrm{tr}\left(L\left(s,\widetilde{V},W^-\right)
\right)_-^{\sigma+1/2}\,\mathrm{d}s
\end{equation}
for all $\sigma\ge3/2$. Since the potential is independent of the transverse variables, the spectrum of $L\left(s, \widetilde{V},W^-\right)$ is easy to find: its eigenvalues are
\begin{equation}
\label{eq3.12}
\mu_{k,m} -\left\|1+f|\kappa_1|\right\|_{\infty}^2
\left(W^-(s)+\|\widetilde{V}(s,\cdot)\|_\infty\right)\,, \quad k,m=1,2,\ldots\:,
\end{equation}
where $\mu_{k,m}$ are the eigenvalues of $-\Delta_D^{D_{f(s)}}$. The spectrum of the last named operator is well known \cite{BEK96}: a complete set of eigenfunctions
can be represented in the generalized spherical coordinates, $r$ and $d-2$ angles $\Theta=(\theta_1,\ldots,\theta_{d-3}, \varphi)$, as
$$
\phi_{k,m}(r,\Theta)= r^{1-\frac{d-1}{2}}J_{k+\frac{d-3}{2}}(\kappa_{k,m}r) Y_{k+\frac{d-1}{2},m}(\Theta)\,,
$$
where $J_{j+\frac{d-3}{2}}$ are Bessel functions and $Y_{j+\frac{d-1}{2},m}$ hyperspherical harmonics. Since the functions $\phi_{k,m}$ have to satisfy Dirichlet condition at the boundary of the $(d-1)$-dimensional disc, the factors  $\kappa_{k,m}$ are determined by the requirement
$$
J_{k+\frac{d-3}{2}}\big(\kappa_{k,m}f(s)\big)=0\,,
$$
which yields
\begin{equation}
\label{eq3.13}
\mu_{k,m}=\kappa^2_{k,m}= \left(\frac{j_{k+(d-3)/2,m}}{f(s)}\right)^2\,,\quad k,m=0,1,2,\ldots\,,
\end{equation}
where $j_{l,m}$ is the $m$-th positive zero of the Bessel function $J_l$. Using the unitary equivalence of operators
$H_\Omega$ and $H_0$ in combination with the relations  (\ref{eq3.8})--(\ref{eq3.13}) we arrive at the sought conclusion.
\end{proof}

Note that the result given in Theorem\ref{thm: gen bound} extends easily to regions with one-sided cusps in analogy with the claim of Corollary~\ref{2cor}.

\section{\mbox{Twisted cusps of non-circular cross section in $\mathbb{R}^3$}}
\label{s: twisted}
\setcounter{equation}{0}

Let us now look into another type of nontrivial cusp geometry. As before we will suppose that it cross section change along the curve playing role of the axis, however, now we allow it to be non-circular. Consider an open connected set $\omega_0\subset \mathbb{R}^2$ and a positive function $f:\: \mathbb{R} \to \mathbb{R}$ satisfying the condition (\ref{Omega-cusp}), and set
\begin{equation}
\label{eq4.1}
\omega_s :=f(s)\omega_0\,,
\end{equation}
where we use the conventional shorthand $\alpha A:= \{(\alpha x,\alpha y):\,(x,y)\in A\}$ for $\alpha>0$ and $A\subset \mathbb{R}^2$. Using (\ref{eq4.1}) we define a straight cusped region determined by $\omega_o$ and the function $f$ as  $\Omega_0:= \left\{ (s,x,y):\: s\in\mathbb{R},\,(x,y)\in\omega_s\right\}$ with $\omega_s$.

In the next step we twist the region. We fix a $C^1$-smooth function $\theta:\mathbb{R}\rightarrow\mathbb{R}$ with bounded derivative, $\|\dot \theta\|_\infty <\infty$, and introduce the region $\Omega_\theta$ as the image
\begin{equation}
\label{eq4.2}
\Omega_\theta:=\mathfrak{L}_\theta(\Omega_0)\,,
\end{equation}
where the map $\mathfrak{L}_\theta:\,\mathbb{R}^3\to \mathbb{R}^3$ is given by
\begin{equation}
\label{eq4.3}
\mathfrak{L}_\theta(s,x,y):= \left(s,x\cos\theta(s)+y\sin\theta(s),-x\sin\theta(s)+y\cos\theta(s)\right)\,.
\end{equation}
We are interested primarily in the situation when the region is twisted, that is
 \begin{enumerate}[(i)]
 \setlength{\itemsep}{0pt}
 \item the function $\theta$ is not constant\,,
 \item $\omega_0$ is not rotationally symmetric with respect to the origin in $\mathbb{R}^2$.
 \end{enumerate}
If the first condition is not valid we have a straight region with the cross section rotated by a fixed angle, $\omega_{0,\theta}:= \left\{x\cos\theta+y\sin\theta ,-x\sin\theta+y\cos\theta\,:\,(x,y)\in\omega_0 \right\}$, while if $\omega_0$ is rotationally symmetric (i.e. $\omega_{0,\theta} = \omega_0$ up a set of zero capacity for any $\theta\in(0,2\pi)$) the choice of $\theta$ does not matter.

To formulate the result of this section, we need a few more preliminaries. First of all, we introduce $\varrho := \sup_{(x,y)\in\omega_0}\sqrt{x^2+y^2}$ and assume that
\begin{equation}
\label{eq4.4}
\varrho\|f\dot{\theta}\|_\infty<1\,.
\end{equation}
Next we we set $\widetilde{V}(s,x,y) := V(\mathfrak{L}_\theta(s,x,y))$ in analogy with the corresponding definitions in the previous sections, and finally, we introduce the operator
 $$
L_\mathrm{trans} := -i\left(x\frac{\partial}{\partial y}-y\frac{\partial}{\partial x}\right)\,,\quad \mathrm{Dom} \big(L_\mathrm{trans} \big) = \mathcal{H}^1_0 (\omega_0)\,,
 $$
of the angular momentum component canonically associated with rotations in the transverse plane. Now we are ready to state the result.

\begin{theorem} \label{twist bound}
Let $H_{\Omega_\theta}$ be the operator (\ref{cusp-schroed}) referring to the region $\Omega_\theta$ defined by (\ref{eq4.2}) and (\ref{eq4.3}) with a potential $V\ge 0$ which is bounded and measurable. Under the assumption (\ref{eq4.4}) the negative spectrum of $H_{\Omega_\theta}$ the inequality
$$
\mathrm{tr}\big(H_D^{\Omega_\theta}\big)_-^\sigma \le L_{\sigma,1}^{\mathrm{cl}}
\left(1-\varrho\|f\dot{\theta}\|_{\infty}\right)^\sigma \int_{\mathbb{R}}\sum_{j=1}^\infty \left(-\frac{\lambda_{0,j}(s)}{f^2(s)} +\frac{\|\widetilde{V}(s,\cdot)\|_\infty}{1-\varrho\|f\dot{\theta}\|_\infty}
\right)_+^{\sigma+1/2}\,\mathrm{d}s
$$
holds true for $\sigma\ge3/2$, where $L_{\sigma,1}^{\mathrm{cl}}$ is the constant (\ref{LTconstant}) and
$\lambda_{0,j}(s),\, j=1,2,\ldots,\,$ are the eigenvalues of the operator
$$
H_{f,\theta}(s):=-\Delta_D^{\omega_0}+f^2(s)\dot{\theta}^2(s) L_\mathrm{trans}^2
$$
defined on the domain $\mathcal{H}^2_0 (\omega_0)$ in $L^2(\omega_0)$.
\end{theorem}
\begin{proof}[Proof:]
As before we employ suitable curvilinear coordinates, this time to ``untwist'' the region. We define a unitary operator from $L^2(\Omega_\theta)$ to $L^2(\Omega_0)$ by $U_\theta\psi:= \psi\circ\mathfrak{L}_\theta$ which allows us to pass from $H_{\Omega_\theta}^D$ to the operator
 $$ 
H_0:=U_\theta\left(H_{\Omega_\theta}^D\right)U^{-1}_\theta
 $$ 
in $L^2(\Omega_0)$. From \cite{KZ11} we know that $H_0$
is the self-adjoint operator associated with the quadratic form
$$
Q_0:\: Q_0[\psi]:=\|\partial_s\psi+i\dot{\theta} L_\mathrm{trans} \psi\|^2 +\|\nabla_\mathrm{trans}\psi\|^2
-\int_{\Omega_0} \big(\widetilde{V}|\psi|^2\big) (s,x,y)\,\mathrm{d}s\,\mathrm{d}x\,\mathrm{d}y
$$
defined on $\mathcal{H}_0^1$, where $\nabla_\mathrm{trans}:= \left( \partial_x,\partial_y\right)$ and the norms refer to $L^2(\Omega_0)$. In order to estimate the form we note that
 $$
|L_\mathrm{trans}\phi|\le \varrho f(s)|\nabla_\mathrm{trans}\phi|
$$
holds for any function $\phi\in\mathcal{H}_0^1(\omega_s)$, hence using Cauchy-Schwarz we get
$$
2\,\mathrm{Re}\left|\int_{\Omega_0} \dot{\theta}(s) \big(\partial_s\psi\, \overline{L_\mathrm{trans}\psi}\big) (s,x,y) \,\mathrm{d}s\,\mathrm{d}x\,\mathrm{d}y\right| \le
\varrho\|f\dot{\theta}\|_\infty(\|\partial_s\psi\|^2 +\|\nabla'\psi\|^2)\,,
$$
where the last two norm refer again to $L^2(\Omega_0)$, which in turn yields
\begin{eqnarray}
\lefteqn{Q_0[\psi] \ge \big(1-\varrho\|f\dot{\theta}\|_{\infty}\big) \big(\|\nabla\psi\|^2+
\|\dot{\theta}L_\mathrm{trans}\psi\|^2 \big)} \nonumber \\[.5em] && \label{eq4.5} \hspace{1em}
-\int_{\Omega_0}\|\widetilde{V}(s,\cdot)\|_\infty|\psi (s,x,y)|^2\, \mathrm{d}s\,\mathrm{d}x\,\mathrm{d}y\,.
\end{eqnarray}
Introducing thus the operator
$$
H^-_0=-\Delta_D^{\Omega_0}+\dot{\theta}^2(s) L_\mathrm{trans}^2 -\frac{1}{1-\varrho\|f\dot{\theta}\|
_\infty}\|\widetilde{V}(s,\cdot)\|_\infty
$$
with the domain $\mathcal{H}_0^2(\Omega_0)$, we get from (\ref{eq4.5}) the following lower bound,
\begin{equation}
\label{eq4.6}
H_0\ge \big(1-\varrho\|f\dot{\theta}\|_\infty \big)H_0^-\,,
\end{equation}
which makes sense in view of the condition (\ref{eq4.4}); by
minimax principle it is then enough to establish a bound to the negative spectrum of the operator $H_0^-$. For any $u\in\,C_0^\infty(\Omega_0)$ we can write
\begin{eqnarray*}
\lefteqn{\|\nabla\,u\|^2+\|\dot{\theta}L_\mathrm{trans}u\|^2-
\frac{1}{1-\varrho\|f\dot{\theta}\|_\infty}\int_{\Omega_0}
\|\widetilde{V}(s,\cdot)\|_\infty|u(s,x,y)|^2\,\mathrm{d}s\, \mathrm{d}x\,\mathrm{d}y} \\ && =
\int_{\Omega_0}|\partial_s u(s,x,y)|^2\, \mathrm{d}s\,\mathrm{d}x
\,\mathrm{d}y +\int_{\mathbb{R}}\, \mathrm{d}s\int_{\omega_s} \bigg(|\partial_x u(s,x,y)|^2 +|\partial_y u(s,x,y)|^2
\\ &&
\hspace{1em} +\dot{\theta}^2(s) \big|(L_\mathrm{trans}u)(s,x,y)\big|^2-
\frac{1}{1-\varrho\|f\dot{\theta}\|_\infty} \|\widetilde{V}(s, \cdot)\|_\infty|u(s,x,y)|^2\bigg)\,\mathrm{d}x\, \mathrm{d}y
\\ &&
\ge\int_{\Omega_0}|\partial_s u(s,x,y)|^2\,\mathrm{d}s\,\mathrm{d}x\,
\mathrm{d}y+\int_{\mathbb{R}}\left\langle\,H(s,\widetilde{V})\,u(s,\cdot),\,
u(s,\cdot)\right\rangle_{L^2(\omega_s)}\,\mathrm{d}\,s\,,
\end{eqnarray*}
where the norm without a label refer again to $L^2(\Omega_0)$ and $H(s,\widetilde{V})$ is the negative part of two-dimensional Schr\"{o}dinger operator
\begin{equation}
\label{eq4.7}
-\Delta_D^{\omega_s}+\dot{\theta}^2(s) L_\mathrm{trans}^2 -\frac{1}{1-\varrho\|f\dot{\theta}\|_\infty}\|\widetilde{V}(s, \cdot)\|_\infty
\end{equation}
defined on $\mathcal{H}_0^2(\omega_s)$. The next step is analogous to what we did in the proofs of Theorems~\ref{thm: 2D bound} and \ref{thm: gen bound}: we extend the operator $H(s,\widetilde{V})$ to the whole $\mathbb{R}^2$ by regarding it as a direct sum with zero component in $C_0^\infty\left(\mathbb{R}^2\setminus \overline{\omega_s}\right)$. For a function $g=u+v$ with $u\in\,C_0^\infty(\Omega_0)$ and $v\in\,C_0^\infty(\widehat{\Omega}_0)$ extended
by zero in the complement regions in $\mathbb{R}^3$ we have thus the inequality
\begin{eqnarray*}
\lefteqn{\|\nabla\,u\|^2 + \|\nabla\,v\|^2_{L^2(\widehat{\Omega}_0)} +\|\dot{\theta}L_\mathrm{trans}u\|^2} \\ &&
\hspace{1.5em} -\frac{1}{1-\varrho\|f\dot{\theta}\|_\infty}\int_{\Omega_0}
\|\widetilde{V}(s,\cdot)\|_\infty|u(s,x,y)|^2\,\mathrm{d}s\, \mathrm{d}x\,\mathrm{d}y
\\ &&
\ge\int_{\Omega_0}|\partial_s g(s,x,y)|^2\,\mathrm{d}s\,\mathrm{d}x\,
\mathrm{d}y+\int_{\mathbb{R}}\left\langle\,H(s,\widetilde{V})\,u(s,\cdot),\,
u(s,\cdot)\right\rangle_{L^2(\mathbb{R}^2)}\,\mathrm{d}s\,,
\end{eqnarray*}
valid for all $g\in\,C_0^\infty\left(\mathbb{R}^3 \backslash \partial\Omega_0\right)$. Its left-hand side of is the form associated with the operator $H_0^-\oplus \big(-\Delta_D^{\widehat{\Omega}_0} \big)$ while the right-hand side is associated with the operator $-\partial^2_s\otimes\, I_{L^2(\mathbb{R}^2)}+H(s,\widetilde{V})$ defined on the enlarged domain $\mathcal{H}^1\left(\mathbb{R}, L^2(\mathbb{R}^2) \right)$. Using the positivity of $-\Delta_D^{\widehat{\Omega}_0}$ we get
 $$
\mathrm{tr}\,\left(H^-_0\right)_-^\sigma \le \,\mathrm{tr}\left(-\partial^2_s\otimes\,I_{L^2(\mathbb{R}^2)}+
H(s,\widetilde{V})\right)_-^\sigma\,,\quad \sigma\ge 0\,,
 $$
hence the Lieb-Thirring inequality for operator-valued potentials yields
\begin{equation}
\label{eq4.10}
\mathrm{tr}\,\left(H^-_0\right)_-^\sigma\le\,L_{\sigma,1}^{\mathrm{cl}}
\int_{\mathbb{R}} \mathrm{tr} \,H(s,\widetilde{V})_-^{\sigma+1/2}\,\mathrm{d}s\,,\quad \sigma\ge3/2\,,
\end{equation}
with the semiclassical constant $L_{\sigma,1}^{\mathrm{cl}}$ given by (\ref{LTconstant}). Combining thus the unitary equivalence of  $H_D^{\Omega_\theta}$ and $H_0$ with the inequalities (\ref{eq4.6}), (\ref{eq4.10}) and the condition (\ref{eq4.4}) we get
\begin{equation}
\label{eq4.11}
\mathrm{tr}\left(H_D^{\Omega_\theta}\right)_-^\sigma \le\,L_{\sigma,1}^{\mathrm{cl}}
\left(1-\varrho\|f\dot{\theta}\|_{\infty}\right)^\sigma \int_{\mathbb{R}} \mathrm{tr}\,H(s,\widetilde{V})_-^{\sigma+1/2}\,
\mathrm{d}s \quad\mathrm{for}\quad \sigma\ge3/2\,.
\end{equation}
It remains to determine the eigenvalues of the operator (\ref{eq4.7}). Since the potential in it is independent of the transverse variables, it is easy to see that they are
 $$
\frac{\lambda_{0,j}(s)}{f^2(s)}-\frac{1}{1-\varrho\|f\dot{\theta}\|_\infty} \|\widetilde{V}(s,\cdot)\|_\infty\,,\quad j=1,2,\ldots\,,
 $$
where $\lambda_{0,j}(s),\,j=1,2,\ldots$ are the eigenvalues of the operator $H_{f,\theta}(s)$ defined in the theorem. Combining this result with (\ref{eq4.11}) we conclude the proof.
\end{proof}

Once more, the result of Theorem~\ref{twist bound} can be easily extended to twisted regions with one-sided cusps in analogy with the claim of Corollary~\ref{2cor}.

\begin{remark}
{\rm
Similarly as in Remark~\ref{thickend}a we have not used the condition (\ref{Omega-cusp}) which makes the result applicable to twisted tubular regions which do not shrink to zero at infinity as long as we are interested in the spectrum below $\inf \sigma_\mathrm{ess}(H_{\Omega_\theta})$. Note that the latter quantity depends on the functions $f$ and $\theta$. For instance, for a constant $f$ and noncircular $\omega$ an asymptotically constant and positive $\theta$ pushes the threshold up in comparison with an untwisted tube. A local slowdown of the twist then gives rise to a discrete spectrum \cite{EK05} which may be infinite if $\theta$ is not compactly supported; its accumulation rate depends on the asymptotic behavior of $\theta$ -- cf.~\cite{BKRS09}. Similar effects may be expected here if the effective attraction due to twist slowdown is replaced by the increased thickness at the cusp base.
}
\end{remark}


\bigskip

\begin{flushleft}

Pavel Exner and Diana Barseghyan

\smallskip

Doppler Institute for Mathematical Physics and Applied
Mathematics \\ B\v{r}ehov\'{a} 7, 11519 Prague \\ and  Nuclear
Physics Institute ASCR \\ 25068 \v{R}e\v{z} near Prague, Czechia

\smallskip

Email: exner@ujf.cas.cz, dianabar@ujf.cas.cz

\end{flushleft}

\end{document}